\documentclass[12pt]{article}
\usepackage{amsfonts,amssymb,amsbsy,amsmath,amsthm,enumerate}

\newtheorem{theorem}{Theorem}[section]

\newtheorem{lemma}[theorem]{Lemma}
\newtheorem{follow}[theorem]{Corollary}

\newtheorem{pr}[theorem]{Proposition}
\theoremstyle{definition}

\newcommand{\bel}{\begin{equation} \label}
\newcommand{\ee}{\end{equation}}

\newcommand{\rd}{{\mathbb R}^{2}}
\newcommand{\re}{{\mathbb R}}
\newcommand{\ze}{{\mathbb Z}}
\newcommand{\R}{{\mathbb R}}
\newcommand{\N}{{\mathbb N}}

\def\beq{\begin{equation}}
\def\eeq{\end{equation}}
\newcommand{\bea}{\begin{eqnarray}}
\newcommand{\eea}{\end{eqnarray}}
\newcommand{\beas}{\begin{eqnarray*}}
\newcommand{\eeas}{\end{eqnarray*}}

{

\begin{document}
\begin{center}
{\Large \bf  Discrete Spectrum of Quantum Hall Effect Hamiltonians
II. Periodic Edge Potentials}

\medskip

{\sc Pablo Miranda, Georgi Raikov}

\medskip

\today
\end{center}

\medskip

{\bf Abstract.} {\small We consider the unperturbed operator $H_0
: = (-i\nabla - {\bf A})^2 + W$, self-adjoint in $L^2(\rd)$. Here
${\bf A}$ is a magnetic potential which generates a constant
magnetic field $b>0$, and the edge potential $W = \overline{W}$ is
a ${\mathcal T}$-periodic non-constant bounded function depending
only on the first coordinate $x \in \re$ of $(x,y) \in \rd$.  Then
the spectrum $\sigma(H_0)$ of $H_0$ has a band structure, the band
functions are $b {\mathcal T}$-periodic, and  generically there
are infinitely many open gaps in $\sigma(H_0)$. We establish
explicit sufficient conditions which guarantee that a given band
of $\sigma(H_0)$ has a positive length, and all the extremal
points of the corresponding band function are non-degenerate.
Under these assumptions we consider the perturbed operators
$H_{\pm} = H_0 \pm V$ where the electric potential $V \in
L^{\infty}(\rd)$ is non-negative and decays at infinity. We
investigate the asymptotic distribution of the discrete spectrum
of $H_\pm$ in the spectral gaps of $H_0$. We introduce an
effective Hamiltonian which governs the main asymptotic term; this
Hamiltonian could be interpreted as a 1D Schr\"odinger operator
with infinite-matrix-valued potential. Further, we restrict our
attention on perturbations $V$ of compact support. We find that
there are infinitely many discrete eigenvalues in any open gap in
the spectrum of $\sigma(H_0)$, and the convergence of
these eigenvalues to the corresponding spectral edge is asymptotically Gaussian.}\\



\section{Introduction}
\label{section1} \setcounter{equation}{0}

    The general form of the unperturbed operators we will
    consider in the article, is
    $$
    H_0 = H_0(b,W) : = -\frac{\partial^2}{\partial x^2} +
    \left(-i\frac{\partial}{\partial y} - bx\right)^2 + W(x).
    $$
Here $b>0$ is the constant magnetic field, and $W = \overline{W}
\in L^{\infty}(\re)$ is an electric potential independent of $y$.
The self-adjoint operator $H_0$ is defined initially on
$C_0^{\infty}(\rd)$ and then is closed in $L^2(\rd)$.  Let
${\mathcal F}$ be the partial Fourier transform with respect to
$y$, i.e.
$$
({\mathcal F}u)(x,k) = (2\pi)^{-1/2} \int_\re e^{-iky} u(x,y) dy,
\quad u \in L^2(\rd).
$$
Then we have
$$
{\mathcal F} H_0 {\mathcal F}^* = \int_\re^{\oplus} h(k) dk
$$
where the operator
$$
h(k) : = - \frac{d^2}{dx^2} + (bx-k)^2 + W(x), \quad k \in \re,
$$
is self-adjoint in $L^2(\re)$. Since the multiplier by $x \in \re$ is relatively compact in the sense of the quadratic forms with respect to
$h(0)$, we easily find that $h(k)$ is a Kato analytic family (see e.g. \cite[Theorem XII.10]{RS4}).\\
For $w \in L^2(\re)$ and $k \in
\re$ set $(U_k w)(x) : = w(x-k/b)$. Then $U_k$ is a unitary
operator in $L^2(\re)$, and we have $U_k^* h(k) U_k =
\tilde{h}(k)$ where
$$
\tilde{h}(k) : = - \frac{d^2}{dx^2} + b^2 x^2  +  W(x + k/b),
\quad k \in \re.
$$
Evidently, for each $k \in \re$ the operator $h(k)$ (and, hence,
$\tilde{h}(k)$) has a discrete and simple spectrum. Let
$\left\{E_j(k)\right\}_{j=1}^{\infty}$ be the increasing sequence
of the eigenvalues of $h(k)$ (and, hence, of $\tilde{h}(k)$). The
general Kato analytic perturbation theory (see \cite{K} or \cite{RS4})
implies that $E_j(k)$,
$j \in {\mathbb N}$, are real analytic functions of $k \in \re$.
Since $E_j(k)$ depend on the parameters $b$ and $W$, we will
sometimes write $E_j(k;b,W)$ instead of $E_j(k)$. \\
Even though in some parts of the article we will impose more general conditions on $W$, in our main theorems we will assume that $W$ is a periodic function with
    period ${\mathcal T} > 0$, which is not identically constant.
    The explicit expression for the operator $\tilde{h}(k)$
    implies that all the functions $E_j$, $j \in {\mathbb N}$, are
    periodic functions of period $\tau : = b {\mathcal T}$.
    Set
$$
{\mathcal E}_j^- = \min_{k \in [0,\tau)}E_j(k), \quad {\mathcal E}_j^+
= \max_{k \in [0,\tau)} E_j(k).
$$
Then we have
$$
\sigma(H_0) = \bigcup_{j=1}^{\infty} [{\mathcal E}_j^-, {\mathcal
E}_j^+].
$$
We will call the closed intervals $[{\mathcal E}_j^-, {\mathcal
E}_j^+]$, $j \in {\mathbb N}$, the bands of the spectrum of
$H_0$. Note that if for some $j \in {\mathbb N}$ we have
    \bel{4}
{\mathcal E}_j^+ < {\mathcal E}_{j+1}^-,
    \ee
then the interval $({\mathcal E}_j^+, {\mathcal E}_{j+1}^-)$ is an
open  gap in the spectrum of $H_0$.\\
Further, assume that the perturbative  electric potential $V : \rd \to
\re $ is $\Delta$-compact. A simple sufficient condition which
guarantees the compactness of the operator $V (-\Delta + 1)^{-1}$,
is that $V \in L^{\infty}(\rd)$, and
    \bel{per21}
V(x,y) \to 0 \quad {\rm as} \quad x^2 + y^2 \to \infty.
        \ee
By the diamagnetic inequality, $V$ is then also a relatively
compact perturbation of $H_0$, and, hence,  we have
$$
\sigma_{\rm ess}(H_0 + V) = \sigma_{\rm ess}(H_0) =
\bigcup_{j=1}^{\infty} [{\mathcal E}_j^-, {\mathcal E}_j^+].
$$
 For simplicity, we will consider perturbations of definite
sign. More precisely we will suppose that $V \geq 0$, and will
consider the operators $H_{\pm} : = H_0 \pm V$. Note that in the
case of positive (resp., negative) perturbations, the discrete
eigenvalues of the perturbed operator which may appear in a given
open gap of the spectrum of the unperturbed operator, may
accumulate only to the lower (resp., upper) edge of the gap. \\
Let $T$ be a self-adjoint linear operator in a Hilbert space.
Denote by ${\mathbb P}_{\mathcal O}(T)$ the spectral projection of
$T$ corresponding to the Borel set ${\mathcal O} \subseteq \re$.
For $\lambda > 0$ set
$$
{\mathcal N}_0^-(\lambda) : = {\rm rank}\,{\mathbb P}_{(-\infty,
{\mathcal E}^-_1-\lambda)}(H_-).
$$
Next, fix $j \in {\mathbb N}$ and assume that \eqref{4} holds.
Pick $\lambda \in (0, {\mathcal E}^-_{j+1}-{\mathcal E}^+_j)$, and
set
$$
{\mathcal N}_j^-(\lambda) : = {\rm rank}\,{\mathbb P}_{({\mathcal
E}^+_j, {\mathcal E}^-_{j+1}-\lambda)}(H_-),
$$
$$
{\mathcal N}_j^+(\lambda) : = {\rm rank}\,{\mathbb P}_{({\mathcal
E}^+_j+\lambda, {\mathcal E}^-_{j+1})}(H_+).
$$
The aim of the article is to investigate the asymptotic behaviour
as $\lambda \downarrow 0$ of the functions ${\mathcal
N}_j^{\pm}(\lambda)$. For definiteness, we consider only the
asymptotics of  ${\mathcal N}_j^{+}(\lambda)$ while the
asymptotics of ${\mathcal N}_j^{-}(\lambda)$ could be
considered in a completely analogous manner.\\
The paper is organized as follows. In Section \ref{section2}  we
discuss some properties of the band functions $E_j$, $j \in
{\mathbb N}$, necessary for the formulation and the proofs of our
main results. In particular, we obtain explicit conditions which
guarantee that for a given $j \in {\mathbb N}$ we have ${\mathcal
E}_j^- < {\mathcal E}_j^+$, and, moreover, that all the extrema of
$E_j$ are non-degenerate. These explicit conditions could be of
independent interest for other models and problems involving
similar unperturbed
operators. \\
Section \ref{section3} contains the statements of our main
results. In Theorems \ref{pert1}, \ref{pert2}, and Corollary
\ref{perf1} we introduce several versions of the effective
Hamiltonians which are responsible for the main asymptotic term as
$\lambda \downarrow 0$ of ${\mathcal N}_j^+(\lambda)$, and
establish the corresponding asymptotic bounds. In Theorem
\ref{pert3} we consider compactly supported perturbations $V$ and
prove that if  the spectral gap ${\mathcal E}_j^+ <  {\mathcal
E}_{j+1}^-$ is open, then it contains infinitely many discrete
eigenvalues of $H_+$, and the convergence of these eigenvalues to
the edge
${\mathcal E}_j^+$ is asymptotically Gaussian provided that all the maxima of $E_j$ are non-degenerate.\\
The proofs of our main results could be found in Section
\ref{section4}.

\section{Basic spectral properties of $H_0$}
\label{section2}
\setcounter{equation}{0}
In this section we describe some spectral properties of the unperturbed operator $H_0$ needed for the formulation and the proofs of the main results.\\
First we recall a simple condition on $W$, which guarantees that
\eqref{4} holds true for all $j \in {\mathbb N}$. Note that if
$W=0$, then the eigenvalues $E_j$ are independent of $k$, and
their explicit form is well-known:
$$
E_j(k;b,0) = E_j(b,0) = b (2j-1), \quad k \in \re, \quad j \in {\mathbb N}.
$$
    Set
    $$
    W_- : = {\rm ess \, inf}_{x \in \re} W(x), \quad W_+ : = {\rm ess \, sup}_{x \in \re}
    W(x).
    $$
    By the mini-max principle,
    $$
    b(2j-1) + W_- \leq E_j(k; b, W) \leq b(2j-1) + W_+, \quad \quad k \in \re, \quad j \in {\mathbb N},
    $$
    and, hence,
$$
[{\mathcal E}_j^-, {\mathcal E}_j^+] \subseteq [b(2j-1) + W_-,
b(2j-1) + W_+], \quad j \in {\mathbb N}.
$$
Thus, a sufficient (but not
always necessary) condition which guarantees that \eqref{4} holds
for all $j \in {\mathbb N}$, is
    \bel{5a}
    W_+ - W_- < 2b.
    \ee
Fix $j \in {\mathbb N}$. The asymptotics as $\lambda \downarrow 0$
of ${\mathcal N}_j^{\pm}(\lambda)$ depends crucially on the
 set
$$
{\mathcal M}_j^{\pm} : = \left\{k \in [0,\tau) \, | \, E_j(k) =
{\mathcal E}_j^{\pm}\right\},
$$
and the behaviour of $E_j$ in a vicinity of this set. Even though
we investigate  for definiteness only the asymptotics of
${\mathcal N}_j^+$, here it is convenient to consider both
sets ${\mathcal M}_j^{\pm}$.\\
 First of all, we assume that the band
function
$E_j$  is not identically constant. Corollary \ref{perf2} below contains an explicit sufficient condition for this.\\
 Further, since the functions $E_j$ are periodic, non-constant, and real-analytic, every set $
{\mathcal M}_j^\pm$, $j \in {\mathbb N}$, is non-empty and finite,
i.e.  ${\mathcal M}_j^\pm = \left\{k_{\alpha,
j}^\pm\right\}_{\alpha=1}^{A_j^\pm}$, $A_j^\pm \in {\mathbb N}$.
    Moreover, for each $k_{\alpha, j}^\pm \in {\mathcal M}_j^\pm$ there exists $l =
l(k_{\alpha, j}^\pm) \in {\mathbb N}$ such that
$$
\frac{d^sE_j}{dk^s}(k_{\alpha, j}^\pm) = 0, \quad s=1,\dots,2l-1,
\quad \mbox{and} \quad \mp \frac{d^{2l}E_j}{dk^{2l}}(k_{\alpha,
j}^\pm)
>0.
$$
If $l(k_{\alpha, j}^\pm) = 1$ for some $k_{\alpha, j}^\pm \in
{\mathcal M}_j^\pm$, we say that $k_{\alpha, j}^\pm$ is a
non-degenerate point, and set
    \bel{per22}
    \mu_{\alpha, j}^\pm : = \mp \frac{1}{2} E_j''(k_{\alpha, j}^\pm).
    \ee

    Fix $j \in \N$. Denote by $\pi_j(k)$ is
the orthogonal projection onto ${\rm Ker}\,(h(k)-E_j(k))$.

\begin{lemma}\label{leiwa}
Let $W\in L^{\infty}(\R, \re)$. Fix  $j \in \N$. Then there exists
a real eigenfunction $\psi_j(\cdot;k) \in {\rm Ran}\,\pi_j(k) =
{\rm Ker}\,(h(k)-E_j(k))$ such that
$\|\psi_j(\cdot;k)\|_{L^2(\re)} = 1$, and the mapping
    \bel{april20}
    \R \ni k \mapsto \psi_j(\cdot;k)\in L^2(\R)
    \ee
    is analytic.
\end{lemma}
\begin{proof}
Our argument will follow the main lines of the proof of
\cite[Lemma 2.3 (v)]{iw2}, which on its turn is based on
\cite[Theorem XII.12]{RS4} (see also the original work
\cite{kato2}). Since the coefficients of the differential
 operator $h(k)$ are real, there exists a
real eigenfunction $\psi_j(\cdot;0) \in {\rm Ran}\,\pi_j(0)$ such
that $\|\psi_j(\cdot;0)\|_{L^2(\re)} = 1$. On the other hand,
\cite[Theorem XII.12]{RS4} implies that for $k$ in a complex
vicinity of the real axis, there exists an analytic family of invertible
bounded operators $\omega(k)$ such that
    \bel{april1}
    \omega_j(k) \pi_j(0) = \pi_j(k) \omega_j(k).
    \ee
    Moreover, for real $k$, the operators $\omega_j(k)$ can be
    chosen to be unitary. Following the argument in the proof of
    \cite[Lemma 2.3 (v)]{iw2}, we find that in our case of a
    differential operator with real coefficients, the operator
    $\omega_j(k)$ can be chosen to be real and unitary for real
    $k$. Set
    $$
    \psi_j(\cdot ; k) : = \omega_j(k) \psi_j(\cdot; 0).
    $$
    Evidently, for $k \in \re$, the function $\psi_j(\cdot ; k)$ is real,
    and $\|\psi_j(\cdot;k)\|_{L^2(\re)} = 1$, while \eqref{april1}
    implies that the mapping defined in \eqref{april20} is analytic. \end{proof}

In the sequel we will use the canonical representation
$$
\pi_j(k) = \langle\cdot,\psi_j(\cdot;k)\rangle \psi_j(\cdot;k)
$$
with an eigenfunction $\psi_j(\cdot;k)$ satisfying the properties
described in Lemma \ref{leiwa}. Put
$$
\tilde{\pi}_j(k) : = U_k \pi_j(k) U_k^* ,  \quad k \in \re, \quad j \in \N.
$$
Then we have
    \bel{per101}
\tilde{\pi}_j(k) : = \langle\cdot,\tilde{\psi}_j(\cdot;k)\rangle
\tilde{\psi}_j(\cdot;k),
    \ee
where $\tilde{\psi}_j(\cdot; k) = U_k^* \psi_j(\cdot ; k)$, or in other words,
$$
\tilde{\psi}_j(x; k) = \psi_j(x + k/b;k), \quad x \in \re, \quad k
\in \re, \quad j \in \N.
$$
Evidently, the function $\tilde{\psi}_j(\cdot; k)$ satisfies the equation
    \bel{5}
    -\frac{\partial^2\tilde{\psi}_j}{\partial
x^2}(x;k)+ b^2 x^2\tilde{\psi}_j(x;k)+W(x+k/b)\tilde{\psi}_j(x;k)
= E_j(k)\tilde{\psi}_j(x;k).
    \ee
     Moreover, $\|\tilde{\psi}_j(\cdot; k)\|_{L^2(\re)} = 1$.
\begin{pr} \label{p31}
Let $W = \overline{W} \in C^2(\re) \cap L^{\infty}(\re)$ with $W',
W'' \in L^{\infty}(\re)$. Suppose that $W'(x_0) > 0$ (resp.,
$W'(x_0) < 0$) for some $x_0 \in \re$. Pick $j \in {\mathbb N}$. Then
there exists $b_0 = b_0(W,j)$ such that $b > b_0$ implies
$E_j'(bx_0;b,W) > 0$ (resp., $E_j'(bx_0; b, W) < 0$).
\end{pr}
\begin{proof}
By the Feynman-Hellmann formula we have
    \bel{per1}
E'_j(k;b,W) = \frac{1}{b} \int_\re W'(x+k/b) \tilde{\psi}_j(x;k)^2 dx.
    \ee
Pick $b> 2\|W\|_{L^{\infty}}$ and denote by $\Gamma_j$ the
circle  of radius $b$, centered at $b(2j-1)$. Denote by
$\tilde{h}(b,0)$ the harmonic oscillator $-\frac{d^2}{dx^2} +
b^2x^2$. Then the interior of $\Gamma_j$ contains the eigenvalue
$E_j(k;b,W)$ (resp., $b(2j-1)$) of the operator $\tilde{h}(k;b,W)$
(resp., of $\tilde{h}(b,0)$), while the rest of the spectra of
these operators lie in the exterior of $\Gamma_j$. Then, evidently, \eqref{per1} implies
$$
bE'_j(k;b,W) = {\rm Tr} \left(W'(\cdot +k/b) \tilde{\pi}_j(k)\right) =
$$
$$
\frac{1}{2\pi i} {\rm Tr} \left(\int_{\Gamma_j}W'(\cdot +k/b)
(\tilde{h}(k;b,W) - \omega)^{-1} d\omega\right) =
$$
$$
\frac{1}{2\pi i} {\rm Tr} \left(\int_{\Gamma_j}W'(\cdot +k/b)
(\tilde{h}(b,0) - \omega)^{-1} d\omega\right) -
$$
    \bel{20}
\frac{1}{2\pi i} {\rm Tr} \left(\int_{\Gamma_j}W'(\cdot +k/b)
(\tilde{h}(b,0) - \omega)^{-1} W(\cdot +k/b) (\tilde{h}(k;b,W) -
\omega)^{-1} d\omega\right),
    \ee
the contour $\Gamma_j$ being run over in clockwise direction.
Further, we have
$$
\frac{1}{2\pi i} {\rm Tr} \left(\int_{\Gamma_j}W'(\cdot +k/b)
(\tilde{h}(b,0) - \omega)^{-1} d\omega\right) =
$$
$$
b^{1/2} \int_\re W'(x+k/b) \varphi_j(b^{1/2}x)^2 dx = \int_\re
W'(b^{-1/2}y+b^{-1}k) \varphi_j(y)^2 dy =
$$
    \bel{21}
    W'(b^{-1}k) + \int_\re
(W'(b^{-1/2}y+b^{-1}k)-W'(b^{-1}k)) \varphi_j(y)^2 dy,
    \ee
    where $\varphi_j = \overline{\varphi}_j$ satisfies
    $$
    -\varphi_j''(x) + x^2 \varphi_j(x) = (2j-1) \varphi_j(x), \quad \|\varphi_j\|_{L^2(\re)} = 1.
    $$
    It is well-known that
    $$
    \varphi_j(x) = {\rm H}_{j-1}(x) e^{-x^2/2}, \quad x \in \re, \quad j \in {\mathbb N},
    $$
    where ${\rm H}_l$, $l \in \ze_+$, are appropriately normalized Hermite polynomials.
    Combining \eqref{20} and \eqref{21}, we get
    \bel{22}
E'_j(k;b,W) - \frac{1}{b} W'(b^{-1}k) = \frac{1}{b}(K_1 + K_2)
    \ee
    with
    $$
    K_1 : = - \frac{1}{2\pi i} {\rm Tr} \left(\int_{\Gamma_j}W'(\cdot +k/b)
(\tilde{h}(b,0) - \omega)^{-1} W(\cdot +k/b) (\tilde{h}(k;b,W) -
\omega)^{-1} d\omega\right),
$$
$$
K_2 : = \int_\re (W'(b^{-1/2}y+b^{-1}k)-W'(b^{-1}k))
\varphi_j(y)^2 dy.
$$
It is easy to check that we have
    \bel{23}
    |K_1| \leq c_1 b^{-1}, \quad |K_2| \leq c_2 b^{-1/2},
    \ee
    with
    \bel{per3}
    c_1 : =  \|W\|_{L^{\infty}(\re)}
    \|W'\|_{L^{\infty}(\re)}\left(\sum_{l=1}^{\infty}(2|l-j|-1)^{-2}\right)^{1/2}
    \left(\sum_{l \in {\mathbb N} : l\neq j}(2|l-j|-3/2)^{-2} +
    4\right)^{1/2},
    \ee
   \bel{per4}
    c_2 : = \|W''\|_{L^{\infty}(\re)} \int_\re |y| \varphi_j(y)^2
    dy.
    \ee
    Putting together \eqref{22} and \eqref{23}, we get
    \bel{per2}
    \left| E'_j(k;b,W) - \frac{1}{b} W'(b^{-1}k) \right| \leq c_1
    b^{-2}+c_2 b^{-3/2}.
    \ee
    Bearing in mind that by hypothesis $W'(x_0) > 0$ (resp., $W'(x_0) < 0$), we
    find that if
    \bel{per24}
    b > b_0 : = \max\left\{2\|W\|_{L^{\infty}(\re)}, \left(\frac{c_2 + \sqrt{c_2^2 + 4c_1
    |W'(x_0)|}}{2|W'(x_0)|}\right)^2\right\},
    \ee
then $E_j'(b x_0) > 0$ (resp., $E_j'(b x_0) < 0$).
\end{proof}
Proposition \ref{p31} implies immediately the following
\begin{follow} \label{perf2}
Assume that $W$ satisfies the assumption of Proposition \ref{p31}.
Then for each $j \in {\mathbb N}$ there exists $b_0 = b_0(j, W) > 0$ such that  $b > b_0$ implies that
    \bel{per23x}
\inf_{k \in \re} E_j(k; b, W) < \sup_{k \in \re} E_j(k; b, W).
    \ee
\end{follow}
{\em Remark}: The absolute continuity of the spectrum of the
operator $H_0$ is equivalent to the validity of \eqref{per23x} for
any $j \in {\mathbb N}$. Unfortunately, the constant $c_2$ in
\eqref{per4}, and hence $b_0$ in \eqref{per24} grow unboundedly as
$j \to \infty$ so that Corollary \ref{perf2} only implies that for
any $a \in \re$ there exists $\tilde{b}_0 = \tilde{b}_0(a, W)$
such that the absolute continuity of the spectrum of the operator $H_0(b,W)$ on the interval $(-\infty, a)$ follows from $b > \tilde{b}_0$. \\
Many authors have conjectured the  absolute continuity of the
spectrum of the Landau Hamiltonian $H_0(b,0)$ perturbed by generic
periodic potentials $W : \rd \to \re$ such that the flux of the
magnetic field through the unit cell of the lattice of the periods
of $W$ is $2\pi$-rational (note, however, that this  is evidently
false for constant $W$). This conjecture was  proved only recently
by F. Klopp for a $G_\delta$-dense set of potentials $W$ which
satisfy the rational-flux condition (see \cite{Kl}). If $W$
depends only on $x$,  and is periodic, then it always satisfies
the rational-flux condition. Nevertheless, even in this simpler
situation, there is no general proof of the absolute continuity of
$\sigma(H_0(b, W))$ for non-constant periodic $W$. In
\cite[Theorem 4.0.4, Corollary 4.0.5]{Be}  the absolute continuity
of $\sigma(H_0(b, W))$ is proven under an explicit condition on
the Fourier coefficients of $W$, and a smallness assumption on
$\|W\|_{L^{\infty}(\re)}$; see also the related results in
\cite{dss}. One of the difficulties in the proof of the absolute
continuity of $\sigma(H_0(b, W))$ for general non-constant
periodic $W : \re \to \re$, is related to the fact that we have
    \bel{per30}
\lim_{j \to \infty}\left({\mathcal E}_j^\pm - b(2j-1) - \langle
W\rangle\right) = 0
    \ee
where $\langle W\rangle$ is the mean value of $W$ (see
\cite{KKP}); in particular, $\lim_{j \to \infty}\left({\mathcal
E}_j^+ - {\mathcal E}_j^-\right) = 0$. On the other hand,
\eqref{per30} implies as a by-product that for $j \in {\mathbb N}$
large enough,
inequality \eqref{4} is valid even if \eqref{5a} does not hold true.\\

\begin{pr} \label{p32}
Let $W = \overline{W} \in C^3(\re) \cap L^{\infty}(\re)$ with $W',
W'', W''' \in L^{\infty}(\re)$. Suppose that $W''(x_0) > 0$
(resp., $W''(x_0) < 0$) for some $x_0 \in \re$. Pick $j \in {\mathbb N}$.
Then there exists $b_1 = b_1(W,j)$ such that $b > b_1$ implies
$E_j''(bx_0;b,W) > 0$ (resp., $E_j''(bx_0; b, W) < 0$).
\end{pr}
\begin{proof}
First of all, note that
    $$
    \frac{\partial \tilde{\psi}_j}{\partial k}(x;k) = \frac{{\partial \psi}_j}{\partial k}(x +
    k/b;k)+ \frac{1}{b} \frac{{\partial \psi}_j}{\partial x}(x + k/b;k).
    $$
    Applying Lemma \ref{leiwa}, we conclude that $\frac{\partial \tilde{\psi}_j}{\partial k}(\cdot;k) \in
    L^2(\re)$. Calculating the derivative with respect to $k$ in \eqref{per1}, we get
\bel{per6}
E'_j(k;b,W) = \frac{1}{b^2} \int_\re W''(x+k/b) \tilde{\psi}_j(x;k)^2 dx + \frac{2}{b} \int_\re W'(x+k/b) \frac{\partial \tilde{\psi}_j}{\partial k}(x;k) \tilde{\psi}_j(x;k) dx.
    \ee
    As in the proof of \eqref{per2}, we suppose that $b > 2\|W\|_{L^{\infty}(\re)}$, and find that
    \bel{per7}
    \left|\int_\re W''(x+k/b) \tilde{\psi}_j(x;k)^2 dx - W''(k/b)\right| \leq c_3 b^{-1} + c_4 {b^{-1/2}}
    \ee
    where the constants $c_3$ and $c_4$ are defined by analogy $c_1$ and $c_2$, replacing
    $W'$ by $W''$ in \eqref{per3}, and $W''$ by $W'''$ in
    \eqref{per4}. Further, obviously,
    \bel{per8}
    \left| \int_\re W'(x+k/b) \frac{\partial \tilde{\psi}_j}{\partial k}(x;k) \tilde{\psi}_j(x;k) dx \right| \leq
    \|W'\|_{L^{\infty}(\re)} \left\| \frac{\partial \tilde{\psi}_j}{\partial k}(\cdot;k)\right\|_{L^2(\re)}.
    \ee
    Since the functions $\frac{\partial \tilde{\psi}_j}{\partial k}(\cdot;k)$ and $\tilde{\psi}_j(\cdot;k)$
    are orthogonal in $L^2(\re)$, we find that
    $$
    \frac{\partial \tilde{\psi}_j}{\partial k}(\cdot;k) = (I - \tilde{\pi}_j(k))
    \frac{\partial \tilde{\psi}_j}{\partial k}(\cdot;k),
    $$
    the orthogonal projection $\tilde{\pi}_j(k)$ being defined in  \eqref{per101}.
    Deriving equation \eqref{5} with respect to $k$, we easily obtain
    \bel{per9}
    \frac{\partial \tilde{\psi}_j}{\partial k}(\cdot;k) =
    - \frac{1}{b} (\tilde{h}(k) - E_j(k))^{-1} (I - \tilde{\pi}_j(k))  W'(\cdot + k/b)
    \tilde{\psi}_j(\cdot;k),
    \ee
    and, hence,
    \bel{per10}
    \left\| \frac{\partial \tilde{\psi}_j}{\partial
    k}(\cdot;k)\right\|_{L^2(\re)} \leq \frac{1}{b^2}
    \|W'\|_{L^{\infty}(\re)}.
    \ee
    Putting together \eqref{per6}, \eqref{per7}, \eqref{per8}, and \eqref{per10}, we obtain
    $$
    \left| E''_j(k;b,W) - \frac{1}{b^2} W''(b^{-1}k) \right| \leq
    c_5
    b^{-3}+c_4 b^{-5/2}
    $$
    with $c_5 : = c_3 + 2 \|W'\|^2_{L^{\infty}(\re)}$.
    Therefore  $W''(x_0) > 0$ (resp., $W''(x_0) < 0$), implies $E_j''(b x_0) > 0$ (resp., $E_j''(b x_0) <
    0$),
    provided that
    $$
    b > b_1 : = \max\left\{2\|W\|_{L^{\infty}(\re)}, \left(\frac{c_4 + \sqrt{c_4^2 +
    4c_5
    |W''(x_0)|}}{2|W''(x_0)|}\right)^2\right\}.
    $$
\end{proof}

{\em Remark}: Propositions \ref{p31} -- \ref{p32} show that for large magnetic fields $b$ the band functions
$E_j$, $j \in {\mathbb N}$, behave quite similarly to the edge potential $W$. This behaviour could be considered as semiclassical. \\

The combination of Propositions \ref{p31} -- \ref{p32} easily
yields the following
    \begin{follow} \label{f23}
    Let $W = \overline{W} \in C^3(\re)$ be a ${\mathcal T}$-periodic function
    such that $W'(x) = 0$, $x \in \re$, implies $W''(x) \neq 0$.
    Assume that the sets
    ${\mathcal M}_W^\pm : = \left\{x \in [0,{\mathcal T}) \, | \, W(x)
= W_\pm\right\}$
    consist of $A_W^\pm \in {\mathbb N}$ points. Then for each $j \in {\mathbb
    N}$ there exists $b_2(j, W) > 0$ such that $b > b_2$ implies
    that the set ${\mathcal M}_j^\pm$ contains exactly $A_W^\pm$ points, and all of them are non-degenerate.
    \end{follow}

\section{Main Results}
\label{section3} \setcounter{equation}{0}
    \subsection{Notations. Auxiliary results}
    \label{ss31}
    This subsection contains notations used for the statement of our main theorems,
    and related auxiliary results needed for their proofs. \\
    Let $X_l$, $l=1,2$, be two separable
    Hilbert spaces. By ${\mathcal L}(X_1, X_2)$ (resp., $S_\infty(X_1, X_2)$) we denote the class of
    bounded (resp., compact) linear operators  $T: X_1 \to X_2$, and by $S_p(X_1,
    X_2)$, $p \in [1,\infty)$, the Schatten-von Neumann class of
    operators $T \in S_\infty(X_1, X_2)$ for which $\|T\|_p : =
    \left({\rm Tr}\,(T^* T)^{p/2}\right)^{1/p} < \infty$. If $X_1 = X_2
    = X$, we will write ${\mathcal L}(X)$ and $S_p(X)$ instead of ${\mathcal L}(X, X)$ and $S_p(X,X)$, $p \in
    [1,\infty]$, respectively.
    Let $T = T^* \in S_\infty(X)$. For $s>0$ set
    $$
    n_{\pm}(s; T) = {\rm rank}\,{\mathbb P}_{(s,\infty)}(\pm T);
    $$
    thus $ n_{\pm}(\cdot; T)$ are the counting functions respectively of the
    positive and the negative eigenvalues of $T$.
    Let $T \in S_\infty(X_1, X_2)$. Put
    $$
    n_{*}(s; T) = n_{+}(s^2; T^*T), \quad s > 0;
    $$
    thus $n_{*}(\cdot ; T)$ is the counting function of the singular
    numbers of  $T$.
    We have
    $$
    n_*(s; T) = n_*(s; T^*), \quad s>0.
    $$
    Moreover, if $X_1 = X_2 = X$, and $T=T^*$, we have
    $$
    n_\pm(s; T) \leq n_*(s; T), \quad s>0.
    $$
    Note that the functions $n_\pm$ satisfy Weyl
    inequalities
    \bel{per31}
    n_+(s(1+\varepsilon); T_1) - n_-(s \varepsilon; T_2) \leq
    n_+(s; T_1 + T_2) \leq  n_+(s(1-\varepsilon); T_1) + n_+(s \varepsilon;
    T_2),
    \ee
    with $s>0$ and $\varepsilon \in (0,1)$, while the function $n_*$ satisfies the Ky Fan
    inequalities
    \bel{per32}
    n_*(s(1+\varepsilon); T_1) - n_*(s \varepsilon; T_2) \leq
    n_*(s; T_1 + T_2) \leq  n_*(s(1-\varepsilon); T_1) + n_*(s \varepsilon;
    T_2),
    \ee
    with $s>0$ and $\varepsilon \in (0,1)$. Finally, for each $s>0$ and $p \in [1,\infty)$ we
    have
     \bel{per33}
     n_*(s; T) \leq s^{-p} \|T\|_p^p.
     \ee
    \subsection{Effective Hamiltonians}
    \label{ss32}
    In this subsection, we introduce the effective Hamiltonians
    which under suitable assumptions on $W$ and $V$ govern the
    main asymptotic term as $\lambda \downarrow 0$ of ${\mathcal
    N}_j^+(\lambda)$, and establish the corresponding asymptotic
    bounds.\\
    In what follows, we assume that $V : \rd \to \re$ is Lebesgue
    measurable, and satisfies the estimates
    \bel{per20}
    0 \leq V(x,y) \leq C_0 (1 + |x|)^{-m_1} (1 + |y|)^{-m_2}, \quad
    (x,y) \in \rd,
    \ee
    with some $C_0 \in [0, \infty)$, and $m_l \in (0,\infty)$,
    $l=1,2$. In particular, \eqref{per20} implies that \eqref{per21}
    holds true. Fix $j \in {\mathbb N}$.  Note  that
    \bel{per100}
     \psi_j(x; l\tau + k) =  \psi_j(x - l{\mathcal T}; k), \quad x \in \re, \quad l \in
     {\mathbb Z}, \quad k \in \re,
     \ee
     the eigenfunction $\psi_j(\cdot; k)$ being introduced in Lemma \ref{leiwa}. \\
     Put $A_j^+ : = \# {\mathcal M}_j^+$, ${\mathcal S}_j : =
     \{1,\ldots, A_j^+\}$. Assume that the set ${\mathcal M}_j^+ =
     \left\{k_{\alpha, j}^+\right\}_{\alpha \in {\mathcal S}_j}$
     contains only non-degenerate points $k_{\alpha, j}^+$.
     For $\lambda > 0$ define  ${\mathcal
     G}_1(\lambda) :   l^2({\mathbb Z} \times {\mathcal
     S}_j) \otimes L^2(\re) \to L^2(\rd)$ as the operator with
     integral kernel
     $$
     (2\pi)^{-1/2} V(x,y)^{1/2} \psi_j(x-l{\mathcal T}; k_{\alpha, j}^+)
     e^{i(k+l\tau + k_{\alpha,j}^+)y} \left(\mu_{\alpha, j}^+ k^2 +
     \lambda\right)^{-1/2},
     $$
     with $(l, \alpha) \in {\mathbb Z} \times {\mathcal S}_j$, $k \in
     \re$, and $(x,y) \in \rd$, the quantities $\mu_{\alpha, j}^+ > 0$ being defined in
     \eqref{per22}. It is easy to check that if $V$ satisfies
     \eqref{per20} with $m_1 > 1$, $m_2 > 1$, then ${\mathcal
     G}_1(\lambda) \in   S_2(l^2({\mathbb Z} \times {\mathcal
     S}_j) \otimes L^2(\re); L^2(\rd))$ for any $\lambda > 0$.
     \begin{theorem} \label{pert1}
     Let $W \in L^{\infty}(\re; \re)$ be a ${\mathcal T}$-periodic function.
     Let $V$ satisfy \eqref{per20} with $m_1>1$, $m_2>1$. Fix $j
     \in {\mathbb N}$. Assume that \eqref{4} holds true, and the
     set ${\mathcal M}_j^+$ contains only non-degenerate points.
     Then for each $\varepsilon \in (0,1)$ we have
     \bel{per23}
     n_*(1+\varepsilon; {\mathcal G}_1(\lambda)) + O(1) \leq
     {\mathcal N}_j^+(\lambda) \leq n_*(1-\varepsilon; {\mathcal G}_1(\lambda)) +
     O(1),
     \ee
     as $\lambda \downarrow 0$.
     \end{theorem}
     The proof of Theorem \ref{pert1} can be found in Subsection
     \ref{ss41}.\\
     Our next goal is to give an equivalent formulation of Theorem
     \ref{pert1} in the terms of an explicit effective
     Hamiltonian. Define the ``diagonal" operator $\mu \in {\mathcal L}(l^2({\mathbb Z}
     \times {\mathcal S}_j))$ by
     $$
     (\mu {\bf u})_{l, \alpha} : = \mu_{\alpha, j}^+ u_{l,\alpha},
     \quad l \in {\mathbb Z}, \quad \alpha \in {\mathcal S}_j,
     $$
     where ${\bf u} : = \left\{u_{l,\alpha}\right\}_{(l,\alpha)  \in
     {\mathbb Z} \times {\mathcal S}_j} \in l^2({\mathbb Z}
     \times {\mathcal S}_j)$. On $l^2({\mathbb Z}
     \times {\mathcal S}_j) \otimes H^2(\re)$ define the operator
     $$
     {\mathcal H}_0 : = \mu \otimes \left(-\frac{d^2}{dy^2}\right)
     $$
     self-adjoint in $l^2({\mathbb Z}
     \times {\mathcal S}_j) \otimes L^2(\re)$. Further, define the
     operator ${\mathcal V} \in {\mathcal L}(l^2({\mathbb Z}
     \times {\mathcal S}_j) \otimes L^2(\re))$ by
     $$
     \left({\mathcal V} {\bf w}\right)_{l, \alpha}(y) : = \sum_{m
     \in {\mathbb Z}, \; \beta \in {\mathcal S}_j} {\mathcal
     V}_{l, \alpha; m, \beta}(y) w_{m, \beta}(y), \quad y \in \re,
    $$
    where
    $$
    {\mathcal
     V}_{l, \alpha; m, \beta}(y) : = \frac{1}{2\pi} \int_{\re}
     V(x,y) \psi_j(x-l {\mathcal T}; k_{\alpha, j}^+) \psi_j(x-m {\mathcal T}; k_{\beta,
     j}^+)dx \; e^{-i((l-m)\tau + k_{\alpha, j}^+ - k_{\beta,
     j}^+)y},
     $$
     and ${\bf w} \in l^2({\mathbb Z} \times {\mathcal S}_j) \otimes
     L^2(\re)$. Thus the operator ${\mathcal H}_0 - g {\mathcal
     V}$ with $g \geq 0$, self-adjoint on ${\rm Dom}({\mathcal
     H}_0)$, can be interpreted as a Schr\"odinger operator on the
     real line with infinite-matrix-valued attractive potential $-g{\mathcal
     V}$, and a coupling constant $g \geq 0$. \\
     Applying the Birman-Schwinger principle and the inverse
     Fourier transform with respect to $k \in \re$, we easily find
     that Theorem \ref{pert1} yields the following
     \begin{follow} \label{perf1}
     Under the hypotheses of Theorem \ref{pert1} we have
     $$
     {\rm rank}\,{\mathbb P}_{(-\infty, -\lambda)}({\mathcal H}_0 - (1-\varepsilon) {\mathcal
     V}) + O(1) \leq {\mathcal N}_j^+(\lambda) \leq
     {\rm rank}\,{\mathbb P}_{(-\infty, -\lambda)}({\mathcal H}_0 - (1+\varepsilon) {\mathcal
     V}) + O(1),
     $$
     as $\lambda \downarrow 0$, for any $\varepsilon \in (0,1)$.
     \end{follow}
     Assuming a somewhat faster decay of $V$ as $y \to \infty$,
     we can obtain an asymptotic estimate similar to \eqref{per23}
     involving an operator which is simpler  than ${\mathcal G}_1(\lambda)$.
     Define ${\mathcal G}_2 : l^2({\mathbb Z} \times {\mathcal S}_j) \to L^2(\rd)$ as the operator with integral kernel
     $$
     \left({\mu_{\alpha, j}^+}\right)^{-1/4} V(x,y)^{1/2} \psi_j(x-l {\mathcal T}; k_{\alpha, j}^+)e^{iy(l\tau + k_{\alpha, j}^+)},
     \quad (l, \alpha) \in {\mathbb Z} \times {\mathcal S}_j, \quad (x,y) \in \rd.
     $$
    Again, if $V$ satisfies \eqref{per20} with $m_1 > 1$, $m_2 > 1$, then
    ${\mathcal G}_2 \in S_2( l^2({\mathbb Z} \times {\mathcal S}_j); L^2(\rd))$.
     \begin{theorem} \label{pert2}
     Let $W \in L^{\infty}(\re; \re)$ be a ${\mathcal T}$-periodic function. Let $V$ satisfy \eqref{per20} with $m_1 > 1$ and $m_2 > 3$.
     Fix $j \in {\mathbb N}$ and assume \eqref{4}. Then for each $\varepsilon \in (0,1)$ we have
        \bel{per57}
        n_*\left((1+\varepsilon) \sqrt{2\sqrt{\lambda}}; {\mathcal G}_2\right) + O(1) \leq {\mathcal N}_j^+(\lambda)
        \leq n_*\left((1-\varepsilon) \sqrt{2\sqrt{\lambda}}; {\mathcal G}_2\right) + O(1), \quad \lambda \downarrow 0.
        \ee
        \end{theorem}
        The proof of Theorem \ref{pert2} can be found in Subsection \ref{ss42}.

     \subsection{Asymptotic bounds of ${\mathcal N}_j^+(\lambda)$ for compactly supported $V$}
    \label{ss33}
    Theorems \ref{pert1} -- \ref{pert2} can be used for the investigation of the
    asymptotic behaviour as $\lambda \downarrow 0$ of ${\mathcal N}_j^+(\lambda)$
    for a large class of rapidly decaying perturbations $V$. In this subsection we
     concentrate on perturbations of compact support. This choice is motivated by:
    \begin{itemize}
    \item the fact that in this case we prove asymptotically Gaussian (i.e. the fastest known)
    convergence of the discrete eigenvalue of the operator $H_+$ to the edge ${\mathcal E}_j^+$ of the gap in $\sigma(H_0)$
    which is  a non-semiclassical behaviour;
    similar Gaussian convergence has been recently found in \cite{bmr} in the case of monotone step-like
    edge potential $W$ and additional assumptions of the geometry of ${\rm supp}\,V$;
        \item the relation between our results and the numerous recent results on the asymptotics of the
        discrete spectrum for various (electric, magnetic, or geometric) compactly supported perturbations
        of the Landau Hamiltonian (see e.g. \cite{rw, mroz, rozt, proz, pe, proz2});
            \item the possible applications in the mathematical theory of the quantum Hall effect and the
            related spectral theory of random Anderson-type perturbations of $H_0(b;W)$, i. e. operators
            of the form $H_\omega = H_0 + V_\omega$ where $V_\omega({\bf x}) =
\sum_{{\bf m} \in {\mathbb Z}^2} \lambda_{\bf m}(\omega) u({\bf x
} - {\bf m})$, ${\bf x} \in \rd$, $\omega \in \Omega$, $\Omega$ is
a probability space,  $\left\{\lambda_{{\bf
m}}(\omega)\right\}_{{\bf m} \in {\mathbb Z}^2}$ are i.i.d. random
variables, and $u \geq 0$ is the deterministic compactly supported
single-site potential; note that the estimates for the discrete
eigenvalues for compactly supported perturbations of the Landau
Hamiltonian obtained in \cite{rw} have been successfully applied
to the study of various spectral and dynamical properties of
random Anderson-type perturbations of the same operator (see
\cite{chkr, kr, gks, Kl1}).
    \end{itemize}
    In order to formulate our last theorem we need the
    following notations. For $t > 0$ set
    ${\rm Ent}\,(t) : = \min\{l \in {\mathbb N} \, | \, l \geq t\}$.
    Further, let $\Omega \subset \rd$ be an open, bounded, non-empty set.
    Let ${\bf V}(\Omega)$ be the set of the closed vertical intervals ${\mathcal J} \subset \Omega$
    of positive length $|{\mathcal J}|$. Evidently, ${\bf V}(\Omega) \neq \emptyset$.
    Put
    $$
    {\mathcal C}(\Omega) : = \sup_{{\mathcal J} \in {\bf V}(\Omega)}
    \frac{1}{{\rm Ent}\,\left(\frac{2\pi}{b {\mathcal T}|{\mathcal J}|}\right)}.
    $$
    Note that if ${\mathcal J} \in {\bf V}(\Omega)$, then there exists a horizontal
    interval ${\mathcal I}$ of positive length, such that the rectangle
    ${\mathcal I} \times {\mathcal J}$ is contained in $\Omega$.
    \begin{theorem} \label{pert3}
    Let $W \in L^{\infty}(\re; \re)$ be a ${\mathcal T}$-periodic function.
Suppose that $V : \rd \to [0,\infty)$ is a Lebesgue measurable
function such that
     \bel{per59}
     C_- \chi_{\Omega_-}(x,y) \leq V(x,y) \leq  C_+ \chi_{\Omega_+}(x,y), \quad (x,y) \in \rd,
     \ee
     where $\chi_{\Omega_\pm}$ are the characteristic functions of the open, bounded and non-empty sets
     $\Omega_\pm \subset \rd$, and $C_\pm \in (0,\infty)$ are constants.
      Fix $j \in {\mathbb N}$ and assume \eqref{4}. Suppose that the set ${\mathcal M}_j^+$ contains
      only non-degenerate points.  Then  we have
        \bel{per58}
        \frac{\sqrt{2}}{\sqrt{b} {\mathcal T}}{\mathcal C}(\Omega_-) \leq
        \liminf_{\lambda \downarrow 0} |\ln{\lambda}|^{-1/2} {\mathcal N}_j^+(\lambda) \leq
        \limsup_{\lambda \downarrow 0} |\ln{\lambda}|^{-1/2} {\mathcal N}_j^+(\lambda) \leq
        \frac{\sqrt{2}}{\sqrt{b} {\mathcal T}} A_j^+
        \ee
        where, as earlier, $A_j^+  = \# {\mathcal M}_j^+$. In particular, if $A_j^+ = 1$,
        and there exists a closed vertical interval ${\mathcal J} \subset \Omega_-$ of length
        $|{\mathcal J}| \geq \frac{2\pi}{b {\mathcal T}}$ so that ${\mathcal C}(\Omega_-) = 1$, we have
        $$
        \lim_{\lambda \downarrow 0} |\ln{\lambda}|^{-1/2} {\mathcal N}_j^+(\lambda) = \frac{\sqrt{2}}{\sqrt{b} {\mathcal T}}.
        $$
        \end{theorem}
        The proof of Theorem \ref{pert3} can be found in Subsection \ref{ss43}.\\

        {\em Remarks}: (i)  Corollary \ref{f23} guarantees the existence of edge
        potentials $W$ and magnetic fields $b$ for which the set ${\mathcal M}_j^+$
        contains only non-degenerate points, and  $A_j^+ = 1$.
        Thus there exist explicit examples where the assumptions
        of Theorem \ref{pert3} are met.\\
        (ii) Theorem \ref{pert3} implies that every open gap
        $({\mathcal E}_j^+, {\mathcal E}_{j+1}^-)$ contains
        infinitely many discrete eigenvalues of the operator $H_+$
        for generic not identically vanishing decaying
        perturbations $V \geq 0$. By \eqref{per58} the asymptotic
        rate of the convergence of these eigenvalues is not faster
        than Gaussian. \\

        In principle, the analysis of the asymptotic
        behaviour as $\lambda \downarrow 0$ of ${\mathcal
        N}_j^+(\lambda)$ without the non-degeneracy assumption
        concerning the set ${\mathcal
        M}_j^+$ is also feasible but much more complicated from
        technical point of view, so that we omit the details.
        However, we would just like to note that
        $(k-k_{\alpha,j}^+)^{2l} = o((k-k_{\alpha,j}^+)^{2})$,
        as $k \to k_{\alpha,j}^+$,
if $l \in {\mathbb N}$, $l > 1$; hence, the replacement
        of non-degenerate points $k_{\alpha,j}^+
        \in {\mathcal
        M}_j^+$ by degenerate ones  does not decrease the quantity $\liminf_{\lambda
        \downarrow 0} |{\ln{\lambda}}|^{-1/2} {\mathcal
        N}_j^+(\lambda)$ (see below \eqref{per48}, \eqref{per70}, and
        \eqref{per71}). Thus we find that Theorem \ref{pert3}
        implies the following

        \begin{follow} \label{perf3}
        Let $W \in L^{\infty}(\re; \re)$ be a ${\mathcal T}$-periodic function.
    Assume that $V : \rd \to [0,\infty)$ is a Lebesgue measurable
function which satisfies \eqref{per21} and the lower bound in
\eqref{per59}.
      Fix $j \in {\mathbb N}$. Assume that the inequalities ${\mathcal E}_j^- < {\mathcal E}_{j}^+$ and \eqref{4} hold true.  Then
        $$
        0 <
        \liminf_{\lambda \downarrow 0} |\ln{\lambda}|^{-1/2} {\mathcal
        N}_j^+(\lambda).
        $$
         In particular, the open gap
        $({\mathcal E}_j^+, {\mathcal E}_{j+1}^-)$ contains
        infinitely many discrete eigenvalues of the operator
        $H_+$, and the asymptotic convergence of these eigenvalues
        to the edge ${\mathcal E}_j^+$ is not faster than
        Gaussian.
        \end{follow}

\section{Proofs of the Main Results}
\label{section4}
\setcounter{equation}{0}
    \subsection{Proof of Theorem \ref{pert1}}
    \label{ss41}
    The Birman-Schwinger principle entails
    \bel{aug1}
    {\mathcal N}_j^+(\lambda) = n_-(1; V^{1/2} (H_0 - {\mathcal E}_j^+ - \lambda)^{-1} V^{1/2}) + O(1), \quad \lambda \downarrow 0.
    \ee
    Choose $\delta > 0$ so small that the intervals ${\mathcal O}_{l, \alpha}(\delta) : = (l\tau + k_{\alpha, j}^+ - \delta, l\tau + k_{\alpha, j}^+ + \delta)$, $l \in {\mathbb Z}$, $\alpha \in {\mathcal S}_j$, are pairwise disjoint. Set ${\mathcal O}_\delta : = \cup_{(l, \alpha) \in {\mathbb Z} \times {\mathcal S}_j} {\mathcal O}_{l, \alpha}(\delta)$. Introduce the orthogonal projection
    $$
    P_{j,\delta} : = {\mathcal F}^* \int^{\oplus}_{{\mathcal O}_\delta} \pi_j(k) dk \,{\mathcal F}
    $$
    acting in $L^2(\rd)$. Since ${\mathcal E}_j^+$ is not in the spectrum of the operator $H_0$ restricted to $(I-P_{j,\delta}) {\rm Dom}(H_0)$, we find that the operator $V^{1/2} (H_0 - {\mathcal E}_j^+ - \lambda)^{-1} (I-P_{j,\delta}) V^{1/2}$ converges in norm as $\lambda \downarrow 0$ to a compact operator. Therefore, the Weyl inequalities \eqref{per31} easily imply
    $$
    n_+(1 + \varepsilon ; V^{1/2} ({\mathcal E}_j^+ - H_0 + \lambda)^{-1} P_{j,\delta} V^{1/2}) + O(1) \leq
    $$
    $$
    n_-(1; V^{1/2} (H_0 - {\mathcal E}_j^+ - \lambda)^{-1} V^{1/2})  \leq
    $$
    \bel{per48}
    n_+(1 - \varepsilon; V^{1/2} ({\mathcal E}_j^+ - H_0 + \lambda)^{-1} P_{j,\delta} V^{1/2}) + O(1), \quad \lambda \downarrow 0,
    \ee
    with $\varepsilon \in (0,1)$.\\
     For $\lambda > 0$ define $T_1(\lambda) : L^2(\rd) \to L^2({\mathcal O}_\delta)$ as the operator with integral kernel
    \bel{per70}
    (2\pi)^{-1/2}  ({\mathcal E}_j^+ - E_j(k) + \lambda)^{-1/2}  \psi_j(x;k) e^{-iky} V(x,y)^{1/2}, \quad (x,y) \in \rd, \quad k \in {\mathcal O}_\delta.
    \ee
     Then we have
    \bel{per71}
    V^{1/2} ({\mathcal E}_j^+ - H_0 + \lambda)^{-1} P_{j,\delta} V^{1/2} = T_1(\lambda)^* T_1(\lambda),
    \ee
    and hence
     \bel{per49}
     n_+(s^2; V^{1/2} ({\mathcal E}_j^+ - E_j(k) + \lambda)^{-1} P_{j,\delta} V^{1/2}) = n_*(s; T_1(\lambda)) = n_*(s; T_1(\lambda)^*), \quad s>0.
     \ee
     Let  ${\mathcal W} : L^2({\mathcal O}_\delta) \to l^2({\mathbb Z} \times {\mathcal S}_j) \otimes L^2(-\delta, \delta)$ be the unitary operator defined by
    $$
    ({\mathcal W} u)_{l, \alpha}(k) : = u(k + l \tau + k_{\alpha, j}^+), \quad (l, \alpha) \in {\mathbb Z} \times {\mathcal S}_j, \quad k \in (-\delta, \delta),
    $$
    with $u \in L^2({\mathcal O}_\delta)$. Define $T_2(\lambda) : l^2({\mathbb Z} \times {\mathcal S}_j) \otimes L^2(-\delta, \delta) \to L^2(\rd)$, $\lambda > 0$, as the operator with integral kernel
    $$
    (2\pi)^{-1/2} V(x,y)^{1/2}  \psi_j(x-l {\mathcal T}; k+ k_{\alpha, j}^+) e^{i(k + l \tau + k_{\alpha, j}^+)y} ({\mathcal E}_j^+ - E_j(k+ k_{\alpha, j}^+) + \lambda)^{-1/2},
    $$
    where $(l, \alpha) \in {\mathbb Z} \times {\mathcal S}_j$, $k \in (-\delta, \delta)$, $(x,y) \in \rd$.
    By \eqref{per100}, we have $T_2(\lambda) {\mathcal W} = T_1(\lambda)^*$. Therefore,
    \bel{per51}
    n_*(s; T_1(\lambda)^*) = n_*(s; T_2(\lambda)), \quad s>0, \quad \lambda > 0.
    \ee
    Define $T_3(\lambda) : l^2({\mathbb Z} \times {\mathcal S}_j) \otimes L^2(-\delta, \delta) \to L^2(\rd)$, $\lambda > 0$, as the operator with integral kernel
    $$
    (2\pi)^{-1/2} V(x,y)^{1/2}  \psi_j(x-l {\mathcal T};  k_{\alpha, j}^+) e^{i(k + l \tau + k_{\alpha, j}^+)y} (\mu_{\alpha, j}^+ k^2 + \lambda)^{-1/2},
    $$
    with $(l, \alpha) \in {\mathbb Z} \times {\mathcal S}_j$, $k \in (-\delta, \delta)$, $(x,y) \in \rd$. Then
    $$
    \|T_2(\lambda) - T_3(\lambda)\|_2^2 =
    $$
    $$
    (2\pi)^{-1} \sum_{(l,\alpha) \in {\mathbb Z} \times {\mathcal S}_j} \int_{\rd} V(x,y) \int_{-\delta}^{\delta} \left|\psi_j(x-l  {\mathcal T}; k+ k_{\alpha, j}^+)({\mathcal E}_j^+ -E_j(k+ k_{\alpha, j}^+) + \lambda)^{-1/2}  \right.
    $$
    $$
    \left.  -  \psi_j(x-l  {\mathcal T};  k_{\alpha, j}^+) (\mu_{\alpha, j}^+ k^2 + \lambda)^{-1/2}\right|^2 dk \; dx \, dy \leq
    $$
        $$
        \frac{C_1}{\pi} \sum_{\alpha  \in {\mathcal S}_j} \left\{\int_{-\delta}^{\delta} \left|({\mathcal E}_j^+ -E_j(k+ k_{\alpha, j}^+) + \lambda)^{-1/2} - (\mu_{\alpha, j}^+k^2  + \lambda)^{-1/2}\right|^2 dk + \right.
        $$
        \bel{per42a}
        \left. \int_{-\delta}^{\delta} \left(k^{-2} \int_\re
        \left|\psi_j(x; k+ k_{\alpha, j}^+) - \psi_j(x; k_{\alpha, j}^+)\right|^2 dx\right)
        k^2 (\mu_{\alpha, j}^+k^2  + \lambda)^{-1} dk \right\}
        \ee
        where the quantity
        \bel{per47a}
        C_1 : = C_0 \max_{x \in \re} \sum_{l \in {\mathbb Z}} (1 + |x+l  {\mathcal T}|)^{-m_1} \int_\re (1 + |y|)^{-m_2} dy
        \ee
        with $C_0$ being introduced in \eqref{per20}, is finite by $m_1 > 1$ and $m_2 > 1$.
        Since
        $$
        ({\mathcal E}_j^+ -E_j(k+ k_{\alpha, j}^+) + \lambda)^{-1/2} - (\mu_{\alpha, j}^+k^2  + \lambda)^{-1/2} =
        $$
        $$
        \frac{E_j(k+ k_{\alpha, j}^+) - {\mathcal E}_j^+ + \mu_{\alpha, j}^+k^2}{\sqrt{({\mathcal E}_j^+ -E_j(k+ k_{\alpha, j}^+) + \lambda)(\mu_{\alpha, j}^+k^2  + \lambda)}\left(\sqrt{{\mathcal E}_j^+ -E_j(k+ k_{\alpha, j}^+) + \lambda} + \sqrt{\mu_{\alpha, j}^+k^2  + \lambda}\right)},
        $$
        and
        $$
        E_j(k+ k_{\alpha, j}^+) - {\mathcal E}_j^+ + \mu_{\alpha, j}^+k^2 = O(k^3), \quad k \to 0,
        $$
        we find that the first term in the braces at the r.h.s of \eqref{per42a} is uniformly bounded with respect to $\lambda > 0$. Similarly,
        \bel{per43}
        k^2 (\mu_{\alpha, j}^+k^2  + \lambda)^{-1} \leq 1/\mu_{\alpha, j}^+, \quad \lambda >0, \quad k \in \re.
        \ee
        Further, elementary calculations yield
        \bel{per44}
        k^{-2}  \int_\re \left|\psi_j(x; k+ k_{\alpha, j}^+) - \psi_j(x; k_{\alpha, j}^+)\right|^2 dx \leq \int_0^1 \|\pi_j'(ks + k_{\alpha, j}^+)\|^2 ds.
        \ee
        Since the orthogonal projection $\pi_j(k)$ depends analytically on $k$, we find that the combination of \eqref{per43} and \eqref{per44} implies the uniform boundedness with respect to $\lambda > 0$ of the second term in the braces at the r.h.s. of \eqref{per42a}. Therefore \eqref{per42a} yields
        \bel{per45}
        \|T_2(\lambda) - T_3(\lambda)\|_2 = O(1), \quad \lambda \downarrow 0.
        \ee
        Combining \eqref{per32}, \eqref{per33} with $p=2$, and \eqref{per45}, we get
        \bel{per52}
        n_*(s(1+\varepsilon); T_3(\lambda)) + O(1) \leq n_*(s; T_2(\lambda)) \leq n_*(s(1-\varepsilon); T_3(\lambda)) + O(1), \quad \lambda \downarrow 0,
        \ee
        with $s > 0$ and $\varepsilon \in (0,1)$. Finally, define
        $T_4(\lambda) : l^2({\mathbb Z} \times {\mathcal S}_j) \otimes L^2(\re) \to L^2(\rd)$, $\lambda > 0$,
        as the operator with integral kernel
    $$
    (2\pi)^{-1/2} V(x,y)^{1/2}  \psi_j(x-l  {\mathcal T};  k_{\alpha, j}^+) e^{i(k + l \tau + k_{\alpha, j}^+)y} (\mu_{\alpha, j}^+k^2  + \lambda)^{-1/2} \chi_{(-\delta, \delta)}(k),
    $$
    where $(l, \alpha) \in {\mathbb Z} \times {\mathcal S}_j)$, $k \in \re$, $(x,y) \in \rd$, and
    $\chi_{(-\delta, \delta)}$ is the characteristic function of the interval $(-\delta, \delta)$.
    Evidently,
    \bel{per53}
    n_*(s; T_3(\lambda)) = n_*(s; T_4(\lambda)), \quad s > 0, \quad \lambda >0.
    \ee
    At the same time we have
    $$
     \|T_4(\lambda) - {\mathcal G}_1(\lambda)\|_2^2 =
     $$
     $$
     \frac{1}{\pi} \sum_{(l,\alpha) \in {\mathbb Z} \times {\mathcal S}_j} \int_{\rd} V(x,y) \psi_j(x-l  {\mathcal T}; k_{\alpha, j}^+)^2 dx \, dy \int_{\delta}^{\infty}  (\mu_{\alpha, j}^+ k^2 + \lambda)^{-1} dk  \leq \frac{C_1}{\pi \delta} \sum_{\alpha \in {\mathcal S}_j} \frac{1}{\mu_{\alpha, j}^+},
    $$
    the constant $C_1$ being introduced in \eqref{per47a}. Arguing
    as in the derivation of \eqref{per52}, we get
    \bel{per47}
        n_*(s(1+\varepsilon); {\mathcal G}_1(\lambda)) + O(1) \leq n_*(s; T_4(\lambda)) \leq n_*(s(1-\varepsilon); {\mathcal G}_1(\lambda)) + O(1), \quad \lambda \downarrow 0,
        \ee
        with $s > 0$ and $\varepsilon \in (0,1)$. Putting together \eqref{aug1}, \eqref{per48},
         \eqref{per49},  \eqref{per51}, \eqref{per52},
        \eqref{per53}, and \eqref{per47}, we obtain \eqref{per23}.

 \subsection{Proof of Theorem \ref{pert2}}
    \label{ss42}
    We have
    \bel{per54}
    n_*(s; {\mathcal G}_1(\lambda)) = n_+(s^2; {\mathcal G}_1(\lambda) {\mathcal G}_1(\lambda)^*), \quad s>0, \quad \lambda > 0.
    \ee
    The operator $M_1(\lambda) : = {\mathcal G}_1(\lambda) {\mathcal G}_1(\lambda)^* : L^2(\rd) \to L^2(\rd)$ admits the integral kernel
    $$
    \sqrt{V(x,y)V(x',y')} \sum_{(l,\alpha) \in {\mathbb Z} \times {\mathcal S}_j} \frac{e^{-\sqrt{\lambda/\mu_{\alpha, j}^+} |y-y'|}}{2\sqrt{\mu_{\alpha, j}^+ \lambda}}
    e^{i(l\tau + k_{\alpha, j}^+)(y-y')} \psi_j(x-l {\mathcal T}; k_{\alpha, j}^+) \psi_j(x'-l  {\mathcal T}; k_{\alpha, j}^+),
    $$
    with $(x,y), (x',y') \in \rd$. Define $M_2(\lambda) : L^2(\rd) \to L^2(\rd)$ as the operator with integral kernel
    $$
    \sqrt{V(x,y)V(x',y')} \sum_{(l,\alpha) \in {\mathbb Z} \times {\mathcal S}_j} \frac{1}{2\sqrt{\mu_{\alpha, j}^+ \lambda}}
    e^{i(l\tau + k_{\alpha, j}^+)(y-y')} \psi_j(x-l  {\mathcal T}; k_{\alpha, j}^+) \psi_j(x'-l  {\mathcal T}; k_{\alpha, j}^+),
    $$
    with $(x,y), (x',y') \in \rd$. Taking into account \eqref{per20} and the elementary inequalities $0 \leq 1 - e^{-t} \leq t$, $t \geq 0$, we get
    $$
    \|M_1(\lambda) - M_2(\lambda)\|_2^2 \leq
    $$
    $$
     C_0^2 \max_{\alpha \in {\mathcal S}_j} (2\mu_{\alpha, j}^+)^{-2} \times
    $$
    $$
     \int_{\rd}  (1+ |x|)^{-m_1} (1+ |x'|)^{-m_1}
     \left( \sum_{(l,\alpha) \in {\mathbb Z} \times {\mathcal S}_j}
     \left|\psi_j(x-l  {\mathcal T}; k_{\alpha, j}^+) \psi_j(x'-l {\mathcal T}; k_{\alpha, j}^+)\right|\right)^2 dx dx' \times
    $$
    \bel{per54a}
    \int_{\rd} (1+ |y|)^{-m_2} (1+ |y'|)^{-m_2} |y-y'|^2 dy dy'.
    \ee
    Applying the Cauchy-Schwarz inequality, we obtain
    $$
    \int_{\rd}  (1+ |x|)^{-m_1} (1+ |x'|)^{-m_1} \left( \sum_{(l,\alpha) \in {\mathbb Z} \times {\mathcal S}_j}
     \left|\psi_j(x-l  {\mathcal T}; k_{\alpha, j}^+) \psi_j(x'-l {\mathcal T}; k_{\alpha, j}^+)\right|\right)^2 dx dx' \leq
    $$
    \bel{per54b}
    A_j^+ \left(\max_{x \in \re} \sum_{l \in {\mathbb Z}}(1+ |x+l  {\mathcal T}|)^{-m_1}\right)^2 < \infty
    \ee
    since $m_1 > 1$. Similarly,
    \bel{per54c}
    \int_{\rd} (1+ |y|)^{-m_2} (1+ |y'|)^{-m_2} |y-y'|^2 dy dy' < \infty
    \ee
    since $m_2 > 3$. Now, \eqref{per54a} --  \eqref{per54c} imply
    $$
    \|M_1(\lambda) - M_2(\lambda)\|_2 = O(1), \quad \lambda \downarrow 0.
    $$
    Arguing again as in the derivation of \eqref{per52}, we get
    \bel{per55}
    n_+(s(1+\varepsilon); M_2(\lambda)) + O(1) \leq  n_+(s; M_1(\lambda)) \leq
    n_+(s(1-\varepsilon); M_2(\lambda)) + O(1), \lambda \downarrow 0,
    \ee
    with $\varepsilon \in (0,1)$, $s > 0$. Finally,
    $$
    M_2(\lambda) = \frac{1}{2\sqrt{\lambda}} {\mathcal G}_2 {\mathcal G}_2^*, \quad \lambda > 0,
    $$
    and, hence,
    \bel{per56}
    n_+(s^2; M_2(\lambda)) = n_*\left(s\sqrt{2\sqrt{\lambda}}; {\mathcal G}_2\right), \quad s>0, \quad \lambda > 0.
    \ee
    Now the combination of \eqref{per23}, \eqref{per54}, \eqref{per55}, and \eqref{per56}, yields \eqref{per57}.

     \subsection{Proof of Theorem \ref{pert3}}
    \label{ss43}
    In order to prove Theorem \ref{pert3} we need the following
    \begin{lemma} \label{perl1}
    Let $W \in C^1(\re)$ be real-valued periodic function. Then for any bounded interval
    ${\mathcal I} \subset \re$ of positive length, and for any $k_0 \in \re$ we have
    \bel{per57a}
    \lim_{\xi \to \pm \infty} \xi^{-2} \ln{\int_{{\mathcal I}} \psi_j(x-\xi; k_0)^2 dx} = - b.
    \ee
    \end{lemma}
    Relation \eqref{per57a} follows easily from \cite[Theorem 1.1]{KKP}, so that we omit the details. \\

    Let $\Omega \subset \rd$ be an open bounded non-empty set.  Define $T_5(\Omega): l^2({\mathbb Z} \times {\mathcal S}_j) \to L^2(\Omega)$ as the operator with integral kernel
    $$
    \left(\mu_{\alpha, j}^+\right)^{-1/4}   \psi_j(x-l  {\mathcal T};  k_{\alpha, j}^+) e^{i(l \tau + k_{\alpha, j}^+)y}, \quad (l, \alpha) \in {\mathbb Z} \times {\mathcal S}_j, \quad (x,y) \in \Omega.
    $$
    Then \eqref{per59} combined with the mini-max principle implies
    \bel{per61}
    n_*(s; C_- T_5(\Omega_-)) \leq n_*(s; {\mathcal G}_2) \leq n_*(s; C_+ T_5(\Omega_+)), \quad s>0.
    \ee
    Let us prove first the upper bound in \eqref{per58}.
    Since the set $\Omega_+$ is bounded, it is contained in some rectangle ${\mathcal R}_+ : = {\mathcal I}_+ \times {\mathcal J}_+ $
    where ${\mathcal I}_+$ and ${\mathcal J}_+$ are bounded intervals of positive lengths. Evidently,
    \bel{per63}
    n_*(s;  T_5(\Omega_+)) \leq  n_*(s;  T_5({\mathcal R}_+)), \quad s>0.
    \ee
    Let $M_3^+ \in S_{\infty}(l^2({\mathbb Z} \times {\mathcal S}_j))$ be the ``diagonal" operator defined by
    $$
    (M_3^+ {\bf u})_{l,\alpha} = \nu_{l,\alpha}^+ u_{l,\alpha}, \quad (l,\alpha) \in  {\mathbb Z} \times {\mathcal S}_j,
    $$
    where ${\bf u} : = \left\{u_{l,\alpha}\right\}_{(l,\alpha) \in {\mathbb Z} \times {\mathcal S}_j} \in l^2({\mathbb Z} \times {\mathcal S}_j)$, and
    $$
    \nu_{l,\alpha}^+ : = |{\mathcal J}_+|
    \sum_{\beta \in {\mathcal S}_j} \left(\mu_{\beta, j}^+\right)^{-1/2}
    \sum_{m \in {\mathbb Z}}(m^2+1)^{-1} (l^2 + 1) \int_{{\mathcal I}_+} \psi_j(x-l {\mathcal T}; k_{\alpha, j}^+)^2 dx, \quad (l,\alpha) \in  {\mathbb Z} \times {\mathcal S}_j.
    $$
    Applying the Cauchy-Schwarz inequality, we find that
    $T_5({\mathcal R}_+)^* T_5({\mathcal R}_+) \leq M_3^+$,
    which combined with the mini-max principle yields
    $$
    n_*(s\sqrt{2{\sqrt{\lambda}}};  T_5({\mathcal R}_+)) \leq n_+(s^2 2\sqrt{\lambda}; M_3^+) =
    $$
     \bel{per62}
    \# \left\{(l,\alpha) \in  {\mathbb Z} \times {\mathcal S}_j \, | \, \nu_{l,\alpha}^+ > s^2 2\sqrt{\lambda}\right\}, \quad s> 0, \quad \lambda > 0.
    \ee
    Applying Lemma \ref{perl1}, we easily find that
    \bel{per64}
    \lim_{\lambda \downarrow 0}\frac{ \# \left\{(l,\alpha) \in  {\mathbb Z} \times {\mathcal S}_j \, | \, \nu_{l,\alpha}^+ > s \sqrt{\lambda}\right\}}{|\ln{\lambda}|^{1/2}} =
    \frac{\sqrt{2}}{\sqrt{b} {\mathcal T}} A_j^+, \quad s>0.
    \ee
    Combining now \eqref{per57} with the upper bound in \eqref{per61}, \eqref{per63}, \eqref{per62}, and \eqref{per64}, we obtain the upper bound in \eqref{per58}.\\

    Finally, we prove  the lower bound in \eqref{per58}. Let ${\mathcal J}_-$ be a closed
    vertical interval of length $q \in (0,\infty)$, contained in $\Omega_-$.
    Due to the invariance of $H_0$ with respect to $y$-translations,
    we may assume without any loss of generality that there exists a bounded interval ${\mathcal I}_-$ of a positive
     length, such that  ${\mathcal I}_- \times (0,q) \subset \Omega_-$. Set
    $$
    L = L(q) : = {\rm Ent}\,\left(\frac{2\pi}{b {\mathcal T} q}\right) = {\rm Ent}\,\left(\frac{2\pi}{\tau q}\right).
    $$
    Then we have ${\mathcal R}_- : = {\mathcal I}_- \times (0,\frac{2\pi}{\tau L}) \subset \Omega_-$, and therefore
    \bel{per65}
    n_*(s;  T_5(\Omega_-)) \geq  n_*(s;  T_5({\mathcal R}_-)), \quad s>0.
    \ee
    Let $M_3^- \in S_{\infty}(l^2({\mathbb Z}))$ be the ``diagonal" operator defined by
    $$
    (M_3^- {\bf u})_{m} = \nu_{m}^- u_{m}, \quad m \in  {\mathbb Z},
    $$
    where ${\bf u} : = \left\{u_m\right\}_{m \in {\mathbb Z}}$, and
    $$
    \nu_{m}^- : = \frac{2\pi}{\tau L\sqrt{\mu_{1,j}^+}}  \int_{{\mathcal I}_-} \psi_j(x-mL {\mathcal T}; k_{1, j}^+)^2 dx, \quad m \in  {\mathbb Z}.
    $$
    Restricting the operator $T_5({\mathcal R}_-)$ onto the subspace
    $$
    \left\{{\bf u} : = \left\{u_{l,\alpha}\right\}_{(l,\alpha) \in {\mathbb Z} \times
    {\mathcal S}_j} \in l^2({\mathbb Z} \times {\mathcal S}_j) \, | \, u_{l,\alpha} = 0 \quad \mbox{if}
    \quad
    l \not \in L{\mathbb Z} \quad \mbox{or} \quad \alpha \neq 1\right\},
    $$
    applying the mini-max principle, and taking into account that
    $$
    \int_0^{\frac{2\pi}{\tau L}} e^{iL(m-m')\tau y} dy =  \frac{2\pi}{\tau L} \delta_{m, m'}, \quad m, m' \in {\mathbb Z},
    $$
    we easily find that
    \bel{per67}
    n_*\left(s\sqrt{2{\sqrt{\lambda}}};  T_5({\mathcal R}_-)\right) \geq n_+ \left(s^2 2\sqrt{\lambda}; M_3^-\right) =
\# \left\{m \in  {\mathbb Z}  \, | \, \nu_m^- > s^2
2\sqrt{\lambda}\right\}
    \ee
    with $s> 0$ and $\lambda > 0$. Utilizing again Lemma \ref{perl1}, we get
    \bel{per68}
    \lim_{\lambda \downarrow 0}\frac{ \# \left\{m \in  {\mathbb Z}  \, | \, \nu_m^- > s \sqrt{\lambda}\right\}}{|\ln{\lambda}|^{1/2}} =
    \frac{\sqrt{2}}{\sqrt{b} {\mathcal T} L(q)} , \quad s>0.
    \ee
    Putting together \eqref{per57}, the lower bound in \eqref{per61}, \eqref{per65}, \eqref{per67}, and
    \eqref{per68}, and optimizing with respect to $q$, we obtain the lower bound in \eqref{per58}.\\

 {\large \bf Acknowledgements}. The authors were partially
supported by the Chilean Science Foundation {\em Fondecyt} under
Grant 1090467, and by {\em N\'ucleo Cient\'ifico ICM} P07-027-F
``{\em Mathematical Theory of Quantum and Classical Magnetic
Systems"}. \\

    {\sc Pablo Miranda}\\
Departamento de Matem\'aticas,
Facultad de Ciencias,\\
Universidad de Chile, Las Palmeras 3425, Santiago de Chile\\
E-mail: pmirandar@ug.uchile.cl\\

{\sc Georgi Raikov}\\
 Departamento de Matem\'aticas, Facultad de
Matem\'aticas,\\ Pontificia Universidad Cat\'olica de Chile,
Vicu\~na Mackenna 4860, Santiago de Chile\\
E-mail: graikov@mat.puc.cl

\begin{thebibliography} {[10]}
\frenchspacing \baselineskip=12 pt plus 1pt minus 1pt

\bibitem {Be} {\sc  C. B. E Beeken}, {\em Periodic Schr\"odinger
Operators in Dimension Two: Constant Magnetic Fields and Boundary
Value Problems}, Ph.D. Thesis, University of Sussex, 2002.

\bibitem {bmr} {\sc  V. Bruneau, P. Miranda, G. Raikov}, {\em  Discrete spectrum of quantum Hall effect Hamiltonians
I. Monotone edge potentials}, Arxiv Preprint arXiv:1008.5182 (2010) (to appear in Journal of Spectral Theory).

\bibitem {chkr} {\sc J.-M. Combes, P. D. Hislop, F. Klopp, G. Raikov},
{\em Global continuity of the integrated density of states for random Landau Hamiltonians},
Comm. Partial Differential Equations  {\bf 29}  (2004),   1187 -– 1213.



\bibitem {dss} {\sc E. I. Dinaburg,  Ya. G. Sinai,  A. B. Soshnikov},  {\em Splitting of the
low Landau levels into a set of positive Lebesgue measure under
small periodic perturbations},  Comm. Math. Phys.  {\bf 189}
(1997),  559 -- 575.

\bibitem {gks} {\sc F. Germinet, A. Klein, J. H. Schenker}, {\em Dynamical
delocalization in random Landau Hamiltonians},  Ann. of Math. (2)  {\bf 166}  (2007),   215 –- 244.

\bibitem{iw2}{\sc A. Iwatsuka}, {\em Examples of absolutely continuous
Schr\"{o}dinger operators in magnetic fields}, Publ. RIMS, Kyoto
Univ. {\bf 21}  (1985), 385--401.

\bibitem {kato2} {\sc T. Kato}, {\em On the adiabatic theorem of quantum mechanics}, J. Phys. Soc. Japan {\bf 5}
  (1950), 435 -- 439.


\bibitem {K} {\sc T. Kato}, {\em Perturbation Theory for Linear Operators},
  Die Grundlehren der mathematischen Wissenschaften, {\bf 132}
  Springer-Verlag New York, Inc., New York 1966.

 \bibitem {KKP} {\sc M. Klein, E. Korotyaev, A. Pokrovski},
 {\em Spectral asymptotics of the harmonic oscillator perturbed by bounded potentials},
 Ann. Henri Poincar\'e  {\bf 6}  (2005),   747 -- 789.



\bibitem {Kl}  {\sc F. Klopp}, {\em Absolute continuity of the spectrum of a
Landau Hamiltonian perturbed by a generic periodic potential},
Math. Ann.  {\bf 347}  (2010),  675 –- 687.

\bibitem {Kl1}  {\sc F. Klopp}, {\em Lifshitz tails for alloy-type models in a constant magnetic field},  J. Phys. A  {\bf 43}  (2010),  no. 47, 474029, 9 pp.

\bibitem {kr}  {\sc F. Klopp, G. Raikov}, {\em  Lifshitz tails in constant magnetic fields},  Comm. Math. Phys.  {\bf 267}  (2006),  669–701.

\bibitem{mroz}{\sc  M. Melgaard, G. Rozenblum},
{\em Eigenvalue asymptotics for weakly perturbed Dirac and
Schr{\"o}dinger operators with constant magnetic fields of full
rank}, Comm. Partial Differential Equations  {\bf 28} (2003), 697 -- 736.

\bibitem{pe}{\sc  M. Persson},
{\em Eigenvalue asymptotics of the even-dimensional exterior
Landau-Neumann Hamitonian}, Adv. Math. Phys. {\bf 2009} (2009),
Article ID 873704, 15 pp.

\bibitem {proz} {\sc A. Pushnitski, G. Rozenblum},
{\em Eigenvalue clusters of the Landau Hamiltonian in the exterior
of a compact domain}, Doc. Math. {\bf 12} (2007), 569 -- 586.

 \bibitem {proz2} {\sc A. Pushnitski, G. Rozenblum},
   {\em On the spectrum of Bargmann-Toeplitz operators with symbols of a variable
   sign}, ArXiv Preprint arXiv:0912.4486 (2009).


\bibitem{rw}{\sc G. D. Raikov, S. Warzel},  {\em Quasi-classical versus
non-classical spectral asymptotics for magnetic Schr\"odinger
operators with decreasing electric potentials}, Rev. Math. Phys.
{\bf 14} (2002), 1051 -- 1072.

\bibitem {RS4} {\sc M. Reed, B. Simon}, {\it Methods of Modern
    Mathematical Physics IV: Analysis of operators},  Academic Press,
    1978.

    \bibitem {rozt} {\sc G. Rozenblum, G. Tashchiyan}, {\em On the spectral
    properties of the perturbed Landau Hamiltonian},
  Comm. Partial Differential Equations {\bf 33} (2008),
  1048 -- 1081.\\
    \end{thebibliography}
    \end{document}